\documentclass[aps,prl,amsmath,amssymb,twocolumn]{revtex4}
\usepackage{amssymb}
\usepackage{color}
\usepackage{graphicx}
\usepackage{epsfig,amssymb,amsmath,amsthm}

\usepackage{color}
\usepackage{cancel}
\usepackage{tikz-cd}

\theoremstyle{plain}
\newtheorem{thrm}{Theorem}[section]

\newtheorem{prop}[thrm]{Proposition}

\theoremstyle{definition}
\newtheorem{defn}[thrm]{Definition}
\newtheorem{post}{Postulate}[]

\theoremstyle{remark}
\newtheorem*{remark}{Remark}

\newcommand{\pj}[1] {\underbar{$#1$}}

\def\>{\rangle}
\def\<{\langle}

\def\comment#1{}
\def\commentg#1{}
\def\labell#1{\label{#1}}
\def\togli#1{}

\begin{document}
	
\title{The four postulates of quantum mechanics are three}
\author{Gabriele Carcassi$^1$, Lorenzo Maccone$^{2}$ and Christine A.
  Aidala$^1$ }\affiliation{\vbox{1.~Physics Department, University of
    Michigan, 450 Church Street, Ann Arbor, MI 48109-1040,
    United States}\\
  \vbox{2.~Dip.~Fisica and INFN Sez.\ Pavia, University of Pavia, via
    Bassi 6, I-27100 Pavia, Italy}}
\begin{abstract}
  The tensor product postulate of quantum mechanics states that the
  Hilbert space of a composite system is the tensor product of the
  components' Hilbert spaces. All current formalizations of quantum
  mechanics that do not contain this postulate contain some equivalent
  postulate or assumption (sometimes hidden). Here we give a natural
  definition of composite system as a set containing the component
  systems and show how one can logically derive the tensor product
  rule from the state postulate and from the measurement postulate. In
  other words, our paper reduces by one the number of postulates
  necessary to quantum mechanics.
\end{abstract}
	\pacs{}
\maketitle
	
In this paper we derive the tensor product postulate (which, hence,
loses its status of postulate) from two other postulates of quantum
mechanics: the state postulate and the measurement postulate.  The
tensor product postulate does not appear in all axiomatizations of
quantum mechanics: it has even been called ``postulate 0'' in some
literature \cite{zurek}. A widespread belief is that it is a direct
consequence of the superposition principle, and it is hence not a
necessary axiom. {\em This belief is mistaken}: the superposition
principle is encoded into the quantum axioms by requiring that the
state space is a {\em linear} vector space. This is, by itself,
insufficient to single out the tensor product, as other linear
products of linear spaces exist, such as the direct product, the
exterior/wedge product or the direct sum of vector spaces, which is
used in classical mechanics to combine state spaces of linear systems.
These are all maps from linear spaces to linear spaces but they differ
in how the linearity of one is mapped to the linearity of the others
\footnote{For example, in the tensor product
  $a \otimes (b+c) = a \otimes b + a \otimes c$ while in the direct
  product $a \times (b+c) = a \times b + 0 \times c$ where $0$ is the
  zero vector. Also, in the tensor product
  $r (a \otimes b) = (r a) \otimes b = a \otimes (r b)$ while in the
  direct product $r (a \times b) = r a \times r b$, where $r$ is a
  scalar.}.  This belief may have arisen from the seminal book of
Dirac \cite{diracbook}, who introduces tensor products (Chap.~20) by
appealing to linearity. However, he adds the seemingly innocuous
request that the product among spaces be distributive (rather,
bilinear), which is equivalent to postulating tensor products (or
linear functions of them). It is not an innocuous request. For example
it does not hold where the composite vector space of two linear spaces
is described by the direct product, e.g.~in classical mechanics, for
two strings of a guitar: it is not distributive.  [General classical
systems, not only linear ones, are also composed through the direct
product.] Of course, Dirac is not constructing an axiomatic
formulation, so his `sleight of hand' can be forgiven. In contrast,
von Neumann (\cite{vonneumannbook} Chap.~VI.2, also \cite{jauch})
introduces tensor products by noticing that this is a natural choice
in the position representation of wave mechanics (where they were
introduced in \cite{weyl,epr}), and then {\em explicitly postulates}
them in general: ``This rule of transformation is correct in any case
for the coordinate and momentum operators [...]  and it conforms with
the [observable axiom and its linearity principles], we therefore
postulate them generally.''  \cite{vonneumannbook}.  More mathematical
or conceptually-oriented modern formulations
(e.g.~\cite{ozawa,masanes,wootters,nielsenchuang}) introduce this
postulate explicitly.  An interesting alternative is provided in
\cite{ballentinebook,ballentinepaper}: after introducing tensor
products, Ballentine verifies a posteriori that they give the correct
laws of composition of probabilities. Similarly, Peres uses
relativistic locality \cite{peres}. While these procedures seemingly
bypass the need to postulate the tensor product, they do not guarantee
that this is the {\em only} possible way of introducing composite
systems in quantum mechanics. In the framework of quantum logic,
tensor products arise from some additional conditions \cite{matolcsi}
which (in contrast to what is done here) are not connected to the
other postulates.  In \cite{marmo,aerts} tensor products were obtained
by specifying additional physical or mathematical requirements.
\comment{ {\em Il resto della frase da spostare in supp material} }
	
Let us first provide a conceptual overview of our approach. We start from the natural definition of a composite system as the set
of two (or more) quantum systems. The composite system is therefore
made of system $A$ {\em and} (joined with) system $B$ and {\em nothing
  else}. The first key insight is that the first two postulates of
quantum theory (introduced below) already assume that the preparation
of one system is independent from the preparation of another
(statistical independence). In fact, we cannot even talk about a
system in the first place if we cannot characterize it independently.
The second key insight is that, using the law of composition of probabilities of independent events, we can \commentg{easily}find a map $M$ that takes the state of the component systems and gives the composite state for the statistically independent case. These insights are
enough to characterize mathematically the state space of the
composite: the linearity given by the Hilbert space, together with the
fact that the composite system is fully described by the observables
of $A$ and $B$, allows us to extend the construction from the statistically independent composite states to the general case (that includes entangled
states). So the work consists of two interrelated efforts: a physical
argument that starts from the first two postulates and leads to the
necessary existence of the composition map $M$ and its properties
together with a formal argument that shows how $M$ leads to the tensor
product.
	
\comment{mettere nel supplementary: 
}
	
This map $M$ acts on the state spaces of the subsystems. Each pure
state is identified by a ray $\pj{\psi}$, a subspace of the system's
Hilbert space comprising all vectors $\psi$ differing by their
(nonzero) modulus and phase: a one-dimensional complex subspace (a complex plane). In
the same way, constraining the observable $X$ to a particular outcome
value $x_0$ means identifying the subspace comprising all
non-normalized eigenvectors $|x_0\>$ of arbitrary phase such that
$X | x_0 \> = x_0 |x_0\>$. The map $M$ establishes a relationship between the states of the subsystems and the composite, so it is a map between subspaces, not vectors. Therefore, $M$ acts on the projective spaces, where all vectors within the same ray are ``collapsed'' into a single point (i.e.~a quotient space in the equivalence class), removing the unphysical ``overspecification'' of the phase and of the modulus. The physical requirements on $M$ are such that we can find a bilinear map $m$ between vectors that acts consistently with $M$ in terms of subspaces. This map $m$ is the tensor product.

\commentg{testo vecchio: As we need to map states and observables
between subsystems and composite system, the map $M$ derived from the
probability space is really a map between subspaces, namely the
projective spaces. We then use projective geometry to show that the
map $M$ on the projective spaces corresponds to a map $m$ on the
vector spaces which is the tensor product. This is done by showing that preparation independence and statistical
independence of the systems imply three conditions on the map $m$:}

More in detail, the physical requirements of statistical
independence, together with the fact that one can arbitrarily prepare the states of the subsystems, imply three conditions on the map $m$:
(H1)~totality: the map is defined on all states of the subsystems;
(H2)~bilinearity: the map is bilinear thanks to the fundamental
theorem of projective geometry; (H3)~span surjectivity: the span of
the image of the map coincides with the full composite Hilbert space.  We
then prove that, if these three conditions H1, H2 and H3 hold, then
the map $m$ is the tensor product, namely the Hilbert space of the
composite system is the tensor product of the components' Hilbert spaces. The tensor
product ``postulate'' hence loses its status of a postulate. An
overview of all these logical implications is given in
Fig.~\ref{f:fig}. The rest of the paper contains the sketch of this
argument, including all the physical arguments outlined above. The
supplementary material \cite{supp} contains the mathematical details.
	
\begin{figure}[ht]
  \epsfxsize=1.\hsize\leavevmode\epsffile{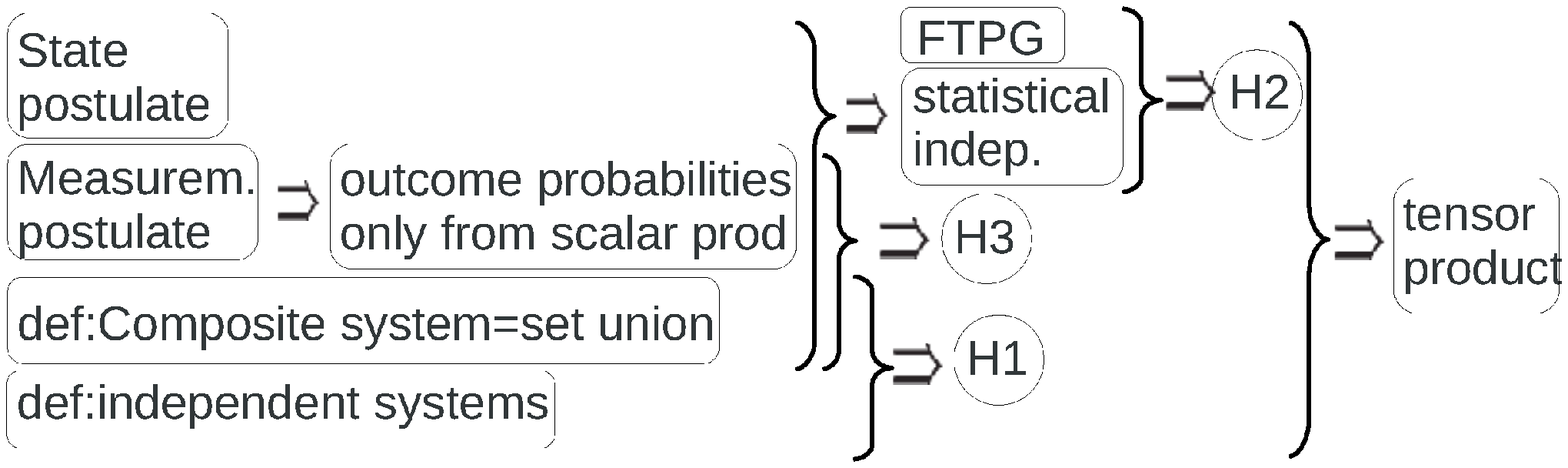}
  \caption{Schematic depiction of the logical implications used in
    this paper. FTPG stands for ``Fundamental Theorem of Projective
    Geometry''.  \label{f:fig}}\end{figure}
	
We start with the axiomatization of quantum mechanics based on the
following postulates
(e.g.~\cite{ozawa,masanes,wootters,nielsenchuang}): (a)~The pure state
of a system is described by a ray $\pj{\psi}$ corresponding to a set
of non-zero vectors $|\psi\>$ in a complex Hilbert space, and the
system's observable properties are described by self-adjoint operators
acting on that space; (b)~The probability that a measurement of a
property $X$, described by the operator with spectral decomposition
$\sum_{x,i}x\frac{|x_i\>\<x_i|}{\<x_i|x_i\>}$ ($i$ a degeneracy index),
returns a value $x$ given that the system is in state $\pj{\psi}$ is
$p(x|\psi)=\sum_i\frac{|\<\psi|x_i\>|^2}{\<x_i|x_i\>\<\psi|\psi\>}$ (Born
rule). (c)~The state space of a composite system is given by the
tensor product of the spaces of the component systems; (d)~The time
evolution of an isolated system is described by a unitary operator
acting on a vector representing the system state,
$|\psi({t})\>=U_{t}|\psi({t}=0)\>$ or, equivalently, by the
Schr\"odinger equation. The rest of quantum theory can be derived from
these axioms. While some axiomatizations introduce further postulates,
we will be using only (a) and (b) to derive (c), so the above are
sufficient to our aims.
	
\togli{This axiomatization implicitly contains a definition of
  ``quantum system'' which is crucial for what follows, so we need to
  clarify the assumptions that it contains. We will use the following
  definition for a quantum
  system\togli{$\stackon[1pt]={\mbox{\tiny
        def}}$}$\stackrel{\mbox{\tiny def}}=${\em ``a quantum degree
    of freedom with $d$ (possibly discrete, or continuous, infinite)
    mutually exclusive (commuting) values for each of its properties.
    Its mathematical description is through a Hilbert space of
    dimension $d$ which contains all the states that describe the
    values of its possible properties. In accordance with the
    postulate (a), these values correspond to a basis of the space,
    given by the eigenvectors of the observable corresponding to that
    property''}. \commentg{We may have to revise to be more clear.
    Where is this used?} }
	
Note that we limit ourselves to kinemati\-cal\-ly-inde\-pen\-dent
systems, where all state vectors $|\psi\>$ in the system's Hilbert
space $\cal H$ describe a valid state, {\em unconditioned on anything
  else}. We call this condition ``preparation independence'' and it
should be noted that the tensor product applies only in this case. For
example, the composite system of two electrons is not the tensor
product, rather the anti-symmetrized tensor product, precisely because
the second electron cannot be prepared in the same state of the first.
We note that restrictions due to superselection rules arise either
from practical (not fundamental) limitations on the actions of the
experimenter \cite{susskind,zanardi,zanardilloyd} or from the use of
ill-defined quantum systems. In the example above, the field is the
proper quantum system and the electrons are its excitations.
\footnote{We emphasize that the kinematic independence is inequivalent
  to dynamical independence (or isolation).  Indeed if two systems
  interact, their interaction may lead to dynamical restrictions in
  the state spaces.  Here we will not consider dynamical evolution,
  which is contained in postulate (d).}
	
The definition of a composite system as containing {\em only} the
collection of the subsystems means that any preparation of both
subsystems independently must correspond to the preparation of the
composite system. Since states are defined by postulate (a) as rays in
the respective Hilbert spaces, there must exist a map
$M : \pj{\cal A} \times \pj{\cal B } \to \pj{\cal C}$ that takes a
pair of states for the subsystems ($\pj{\cal A}$ and $\pj{\cal B }$
represent the projective space, where each point represents all vectors that identify the same state, and the Cartesian product is the set
of all possible pairs) and returns a state in the projective space
$\pj{\cal C}$ for the composite. To visualize the geometrical meaning
of $M$ directly within the Hilbert spaces, given a ray (a complex
plane) in each of $\cal A$ and $\cal B$, $M$ returns a ray (a complex
plane) in $\cal C$.  Our final goal will be to find a map
$m:{\cal A}\times{\cal B}\to{\cal C}$ that acts on vectors in the
Hilbert spaces $\cal A$, $\cal B$ and $\cal C$ consistently with $M$.
Namely, $\pj{m(a,b)}=M(\pj{a},\pj{b})$ where the underline sign
indicates the elements in the projective space. Again geometrically,
$m$ takes a vector in each of $\cal A$ and $\cal B$, and returns a
vector in $\cal C$ and we want this to be consistent with $M$ such
that vectors picked from the same rays will return vectors in the same
ray. We will prove that the map $m$ is the tensor product. We focus on
pure states here: the argument can be extended to mixed states using
standard tools \cite{ballentinebook}.
	
The map $M$ must be injective: as said above, different states of the
subsystems must correspond, by definition of composite system, to
different states of the composite. Moreover, preparation independence
implies that $M$, and hence $m$, must be total maps (condition H1):
each subsystem of the composite system can be independently prepared
and gives rise to a state of the composite.  H1 is not sufficient to
identify the tensor product: by itself it does not even guarantee that
the map $m$ is linear.
	
Postulate (b) contains the connection between quantum mechanics and
probability theory. It must then implicitly contain the axiomatization
of probability, e.g.~see \cite{ballentinepaper,ballentinebook,cox}.
One of the axioms of probability theory (axiom 4 in
\cite{ballentinepaper}) asserts that the joint probability events $a$
and $b$ given $z$ is $p(a\wedge b|z)=p(a|z)\:p(b|z\wedge a)$.
Consider $p(a \wedge b | \psi \wedge b)$ which represents the
probability of measurement outcomes $a$ on system $A$ and $b$ on
system $B$ given that system $A$ was prepared in $\psi$ and system $B$
in $b$.  We have
$p(a \wedge b | \psi \wedge b) = p(a | \psi \wedge b \wedge b) p(b |
\psi \wedge b) = p(a | \psi \wedge b ) p(b | \psi \wedge b)$. The Born
rule tells us that $p(a | \psi \wedge b) = |\<a|\psi\>|^2$ and that
$p(b | \psi \wedge b) = |\<b|b\>|^2 = 1$, where $|a\>$, $|b\>$ are the
normalized eigenstates relative to outcomes $a$ and $b$, and $|\psi\>$
is the normalized state vector. We have:
\begin{align}
  p(a\wedge b|\psi \wedge b)&=p(a|\psi)\\
  p(a\wedge b|a \wedge \phi)&=p(b|\phi)
\end{align}
In other words, since the probability for a measurement on one system
depends only on its pure state, the Born rule requires that the
measurement of one system is independent from the preparation of the
other. We call this property ``statistical
independence'' \footnote{One can also prove that the measurements on
  the components are independent as well (see Supplementary material),
  but we only strictly need preparation here.}.  It characterizes the
map $M$, since $M(\pj{a}, \pj{b})$ corresponds to the composite state
where $A$ and $B$ are prepared in the states $|a\>$ and $|b\>$.
Define $M_b(\pj{a}) = M(\pj{a},\pj{b})$. From the Born rule we find
\begin{align} &\Big|\Big\<M\left(\pj{a},\pj{b}\right)\Big|M\left(\pj{\psi},\pj{b}\right)\Big\>_{\cal C}\Big|^2
	=\Big|\Big\<M_b\left(\pj{a}\right)\Big|M_b\left(\pj{\psi}\right)\Big\>_{\cal C}\Big|^2
	\nonumber \\&
	=\left|\<a|\psi\>_{\cal A}\right|^2
                      \labell{questa},
\end{align}
where the first and second terms contain the inner product in the
composite space $\cal C$. [This is not a new assumption: it follows
from the measurement postulate (b) for the composite system.] This
means that, when one subsystem is prepared in an eigenstate of what is
measured there, the state space of the other is mapped preserving the
square of the inner product.  This implies orthogonality and the
hierarchy of subspaces are preserved through $M_b$, making $M_b$ a
colinear transformation by definition. Geometrically, recall that
$M_b$ maps rays to rays. The fact that $M_b$ is colinear means that it
also maps higher order subspaces to higher order subspaces (lines to
lines, planes to planes, and so on) while preserving inclusion (if a
line is within a plane, the mapped line will be within the mapped
plane). In this case, the fundamental theorem of projective geometry
\cite{fun} applies, which tells us that a unique semi-linear map $m_b$
that acts on the vectors exists in accordance with $M_b$.  Moreover,
conservation of probability further constrains it to be either linear
or antilinear. This tells us that the corresponding $m$ is either
linear or antilinear in the first argument. Namely, if equation
\eqref{questa} holds, then
	\begin{align}
	\<a|\psi\>=\<m(a,b)|m(\psi,b)\>\labell{h2}\;
	\\\mbox{ or }
	\<a|\psi\>=\<m(\psi,b)|m(a,b)\> \labell{h2b}.
\end{align}
In this setting, the antilinear case \eqref{h2b} corresponds to a
change of convention (much like a change of sign in the symplectic
form for classical mechanics) and can be ignored.  Given a Hilbert
space, in fact, we can imagine replacing all vectors and all the
operators with their Hermitian conjugate, mapping vectors into duals
$|\psi\>^\dag=\<\psi|$. These changes would effectively cancel out
leaving the physics unchanged: the two equations $A|w\>=B|z\>$ and
$\<w|A^\dag=\<z|B^\dag$ are equivalent.  (For example, in his first
papers Schr\"odinger used both signs in his equation: effectively
writing {\em two} equivalent equations with complex-conjugate
solutions \cite{sch}. Also Wigner pointed out this equivalence
\cite{wig}, pg.152). We can repeat the same analysis for the second
argument of $m$ to conclude that it is a bilinear map, condition (H2).
	
The last condition, span surjectivity (H3), follows directly from the
definition of a composite system. Since it is composed {\em only} of
the component systems, for any state $c$ of the composite system, we
must find at least one pair $|a\>$, $|b\>$ such that
$p(a\wedge b | c)\neq 0$.  Span-surjectivity follows: namely the span
of the map applied to all states in the component systems spans the
composite system state space. In other words, the composite does not
contain states that are totally independent of (i.e.~orthogonal to)
the states of the components.
	
We have obtained the conditions H1, H2 and H3 from the state postulate
(a), the measurement postulate (b) and the definitions of composite
and independent systems. We now prove that these three conditions
imply that the (up to now unspecified) composition rule $m$ is the
tensor product. More precisely, given a total, span-surjective,
bilinear map $m:{\cal A}\times{\cal B}\to{\cal C}$ that maps the
Hilbert spaces $\cal A$, $\cal B$ of the components into the Hilbert
space $\cal C$ of the composite and that preserves the square of the
inner product, we find that $\cal C $ is equivalent to
$\cal A \otimes \cal B $ and that $m=\otimes$.
	
Proof. Step 1: the bases of the component systems are mapped to a
basis of the composite system. Because of totality property (H1) and
because the square of the inner product is preserved, we can conclude
that, given two orthonormal bases $\{|a_i\>\}\in{\cal A}$ and
$\{|b_j\>\}\in{\cal B}$,
$|\<m(a_i,b_j)|m(a_k,b_\ell)\>|^2=\delta_{ik}\delta_{j\ell}$, namely
$\{|m(a_i,b_j)\>\}$ is an orthonormal set in $\cal C$.  Moreover, the
surjectivity property (H3) guarantees that in $\cal C$ no vectors are
orthogonal to this set. This implies that it is a basis for $\cal C$.
	
Step 2: use the universal property. The tensor product is uniquely
characterized, up to isomorphism, by a universal property regarding
bilinear maps: given two vector spaces $\cal A$ and $\cal B$, the
tensor product ${\cal A}\otimes{\cal B}$ and the associated bilinear
map $T : \cal A \times \cal B \to {\cal A}\otimes{\cal B}$ have the
property than any bilinear map $m:{\cal A}\times{\cal B}\to{\cal C}$
factors through $T$ uniquely.  This means that there exists a {\em
  unique} $I$, dependent on $m$, such that $I \circ T=m$.
In other words, the following diagram commutes:
\begin{center}
		\begin{tikzcd}\mathcal{A}\times\mathcal{B} \arrow[rd, "m"]\arrow[r, "T"] & \mathcal{A}\otimes\mathcal{B}\arrow[d, "I"] \\
			& \mathcal{C}
		\end{tikzcd}
\end{center}
Since $m : \mathcal{A} \times \mathcal{B} \to \mathcal{C}$ is a
bilinear operator (property H2), thanks to the universal property of
the tensor product we can find a unique linear operator
$I : \mathcal{A} \otimes \mathcal{B} \to \mathcal{C}$ such that
$m(a, b) = I(a \otimes b)$. The set $\{ I(a_i\otimes b_j)$
with $|a_i\>$ and $|b_j\>$ orthonormal bases for $\cal A$ and
${\cal B}\}$ forms a basis for $\cal C$, since
$I(a_i\otimes b_j)=m(a_i,b_j)$ and we have shown above that the
latter is a basis.  Thus,
\begin{align} 
  &\<I(a_i\otimes b_j)|I(a_k\otimes b_\ell)
    \>_{\mathcal{C}}=\<m(a_i,
    b_j)|m(a_k,b_\ell)\>_\mathcal{C} \nonumber\\& =
                                                  \delta_{ik}\delta_{j\ell}
                                                  = \<a_i\otimes
                                                  b_j| a_k \otimes b_\ell\>_{\otimes},
	\labell{ecco}\; 
\end{align}
where we used the orthonormality of the bases and the fact that
$|a_i\otimes b_j\>$ is a basis of the tensor product space
${\cal A}\otimes{\cal B}$. Since the function $I$ is a linear
function that maps an orthonormal basis of ${\cal A}\otimes{\cal B}$
to an orthonormal basis of $\cal C$, $I$ is a an isomorphism (a
bijection that preserves the mathematical structure) between
${\cal A}\otimes{\cal B}$ and $\cal C$. As
${\cal C}\cong{\cal A}\otimes{\cal B}$ are isomorphic as Hilbert
spaces, they are mathematically equivalent: $c \in {\cal C}$ and $I^{-1}(c)$ represent the same physical object. In this sense, we can loosely say that $I$ is the identity, as it connects spaces that are physically 
equivalent. So we can directly use the
tensor product to represent the composite state space. This means that
the map $m : \cal A \times \cal B \to \cal C$ is equivalent to the map
$\otimes : \cal A \times \cal B \to \cal A \otimes \cal B$ in the
sense that $m \circ I^{-1} = \otimes$.$\square$

\togli{\comment{Commento finale da aggiungere?  E' interessante
    perche' tutto cio' ci dice che non otteniamo necessariamente il
    prodotto tensore, ma il prodotto tensore modulo una fase locale
    che e' fisicamente irrilevante. Viene da chiedersi se c'e' una
    qualche notazione (evidentemente piu' generale di quella di Dirac)
    che ci permetta di eliminare questa ambiguita' fisicamente
    irrilevante... La notazione di Dirac gia' elimina la necessita' di
    avere una rappresentazione per descrivere gli stati. Pero'
    evidentemente non elimina del tutto la ambiguita' di
    rappresentazione perche' un cambio di coordinate sia sugli stati
    che sugli operatori mi lascia invariata la fisica, ma non lascia
    invariata la notazione di Dirac...  Probabilmente la notazione che
    elimina l'ambiguita' e' quella che utilizza le matrici densita'
    invece dei vettori di stato (vedi Ozawa \cite{ozawa} e Holevo): le
    matrici densita' normalizzate sono phase-independent:
    $\rho=|\psi\>\<\psi|$. Attenzione: l'interpretazione di una
    matrice mista $\rho=\sum_ip_i|\psi_i\>\<\psi_i|$ come ``il sistema
    e' nello stato $|\psi_i\>$ con probabilita' $p_i$'' e' una
    CONSEGUENZA della regola di Born, quindi nella formalizzazione del
    postulato degli stati in termini di matrici densita', questa
    interpretazione non puo' apparire perche' e' una conseguenza di un
    altro postulato.}}

\vskip 1\baselineskip A few comments on the proof: it is based on the
universal property of the tensor product, which uniquely characterizes
it. In step 1 we show that the bilinear map $m$ maps subsystems' bases
into the composite system basis. We also know that there exists a
tensor product map $ {T}=\otimes$ that can compose the vectors in
$\cal A$ and $\cal B$.  In step 2 we use the universal property: since
$m$ is a bilinear map, we are assured that there exists a unique
$I$ such that $I \circ T=m$. Since we show that $I$ is
an isomorphism, then $I$ bijectively maps vectors in $\cal C$
onto vectors in the tensor product space. Namely $m={T}=\otimes$. 
	
We conclude with some general comments. The tensor product structure
of quantum systems is not absolute, but depends on the observables
that are accessible \cite{zanardi,zanardilloyd}. This is due to the
fact that an agent that has access to a set of observables will define
quantum systems differently from an agent that has access to a
different set of observables. Where one agent sees a single system, an
agent that has access to less refined observables (and is then limited
by some superselection rules) can consider the same system as composed
of multiple subsystems. \comment{Supplementary material? }

It has been pointed out before that the quantum postulates are
redundant: in \cite{masanes,hartle} it was shown that the measurement
postulate (b) can be derived from the others (a), (c), (d). Here
instead we have shown how the tensor product postulate (c) can be
logically derived from the state postulate (a), the measurement
postulate (b) and a reasonable definition of independent systems, and
we have described the logical relations among them.  Of course, we do
not claim that this is the {\em only} way to obtain the tensor product
postulate from the others.
	
{\it Acknowledgements:} L.M. acknowledges useful discussions with
M.~Ozawa, P.~Zanardi, S.~Lloyd, D.~Zeh, G.~Auletta, A.~Aldeni and
funding from the MIUR Dipartimenti di Eccellenza 2018-2022 project F11I18000680001,  Attract project through the Eu Horizon 2020 research
and innovation programme under grant agreement No 777222. This material is based upon work supported by the U.S. Department of
Energy, Office of Science, National Quantum Information Science Research
Centers. G.C. and C.A.A. would like to thank M. J. Greenfield for reviewing the
mathematical details and acknowledge funding from the MCubed program
of the University of Michigan. G.C. and C.A.A.'s contribution to this work is part of a larger project, Assumptions of Physics, which aims to identify a handful of physical principles from which the basic laws can be rigorously derived.

	 \section{Supplementary information for ``The four postulates of quantum
  mechanics are three''. }\label{app}
 \subsection{ Mathematical formulation.}
 \setcounter{section}{1}
 \setcounter{equation}{0}
 \counterwithout{equation}{section}
 \renewcommand{\theequation}{S\arabic{equation}}
 
 Here we give the mathematical details of the proof sketched in the
 main paper.
 
 The core idea is that the probability space constrains what the state space for the composite system can be. Therefore we must develop a precise map between events in probability and their correspondents in terms of Hilbert space. Conceptually, the event $X = x_0$, the observable $X$ is equal to the value $x_0$, will correspond to the subspace spanned by all eigenstates of $X$ with eigenvalue $x_0$. The event $\psi$, the system was prepared in state $\psi$, will correspond to the ray (one-dimensional complex subspace) corresponding to the $\psi$ vector. Therefore, in general, all events in probability will correspond to subspaces (of different dimensionality) of the Hilbert space. Projective spaces are the right tool to keep track of subspaces. The proof, then, consists of establishing the correct definitions in the space of events, mapping those into elements of the projective space and then, from the projective space, constructing a map on the vector space directly.

 Let us establish, then, the notation we will be using to distinguish the projective space from
 the Hilbert space itself. If $X$ is a Hilbert space, we denote $\pj{X}$
 the projective space. The projective space is mathematically
 constructed from the Hilbert space by removing the origin and
 quotienting by the equivalence relationship $v \sim \lambda v$, $v\in
 X$ and $\lambda\in\mathbb{C}$. A quantum state is a point in
 projective space. Each point of the projective space is called a ray,
 because for a real vector space it would correspond to a line going through
 the origin, with the origin removed. As we are in a complex space, the
 ray should be thought as a complex plane without the origin, which is
 the space of the vectors reachable from a fixed one through
 multiplication by a complex number. It can also be thought as a
 subspace of dimension one.
 
 Given a vector $v \in X$, we will denote $\pj{v}$ the ray in the
 projective space corresponding to $v$. Note that $\pj{v}$ denotes a
 quantum state, without having picked a modulus or phase. Given two or
 more vectors $v_1, ..., v_n \in X$, the subspace of $X$ they span
 (i.e.~all the vectors reached by linear combinations) is noted by
 $Sp(v_1, ..., v_n)$. Note that this subspace will correspond to a set
 of rays in the projective space, which we note as $\pj{Sp(v_1, ...,
  v_n)}$. Geometrically, this can be thought as the smallest hyper-plane that contains all vectors. Given $v,w \in X$, we can write $P(v|w) = \frac{|\< v | w
  \>|^2}{\< v | v \>\< w | w \>}$ which corresponds to the probability
 of observing $v$ given $w$ was prepared. Note that $P(v|w) = P(\lambda
 v| \mu w)$, with non-null $\lambda,\mu\in{\mathbb C}$, and therefore
 one can write $P(\pj{v} | \pj{w})\equiv P(v|w)$ as a function of the
 rays. Geometrically, this corresponds to the angle between the two complex planes identified by the two vectors.

 \begin{post}\label{post_state}
  The state of a quantum system is described by a ray $\pj{\psi} = \{
  \alpha |\psi\> \, |$ non-null $\alpha \in
  \mathbb{C},|\psi\>\in\mathcal{H}\}$ in a separable complex Hilbert space
  $\mathcal{H}$, and the system's observable properties are described
  by self-adjoint operators acting on that space. 
 \end{post}
 
 \begin{remark}
  All proofs, except one, do not depend on the dimensionality of
  the space. The exception is proposition \ref{prop_fundProj} for which we prove the finite case by induction and then show that it holds in the limit. This would not work in the non-separable case, since the basis would not be countable.
 \end{remark}
 
 \begin{post}\label{post_measurements}
  The probability that a measurement of a property $X$, described by the operator with
  spectral decomposition $X = \sum_{x,i }x \frac{| x_i \> \< x_i |}{\< x_i | x_i \>}$  where $i$ is a degeneracy index, returns a value $x$ depends only on $X$ and on the state of the system $\pj{\psi}$ and is given by $P(x|\pj{\psi})=\sum_i \frac{\<\psi| x_i \> \< x_i |\psi\>}{\< \psi | \psi \>\< x_i | x_i \>}$
  (Born rule).\end{post}
 
 {Given two events $\pj{a}$ and $\pj{b}$, for example $X > x_1$ and $X < x_2$, their conjunction $\pj{a} \wedge \pj{b}$ is the event where both are true, $ x_1 < X < x_2$ in the example. In terms of our Hilbert spaces, both $\pj{a}$ and $\pj{b}$ correspond to subspaces and $\pj{a} \wedge \pj{b}$ is exactly the intersection of the two, which is also a subspace. We should not confuse $\pj{a} \wedge \pj{b}$ with $\pj{a} | \pj{b}$: the first refers to either
  preparation or measurement of both systems  in the respective state, while the second corresponds to preparing one system in one state and measuring the
  other system in the other state \cite{cox}.}
 
 \begin{defn}[Compatible states]\label{def_compatible}
  Let A and B be two systems. Let ${\mathcal{A}}$ and ${\mathcal{B}}$
  be their corresponding state spaces. We say two (pure) states $(\pj{a},
  \pj{b}) \in \pj{\mathcal{A}} \times \pj{\mathcal{B}}$ are compatible
  iff the respective systems can be prepared in such states at the
  same time. Formally, the proposition $\pj{a} \wedge \pj{b}$ is
  possible, which means it does not correspond to the empty set in the
  $\sigma$-algebra of the probability space\footnote{The impossible event is not an event with probability zero, rather it is an event that cannot be created at all. For example, ``the dice shows a number that is even and less than two'', or ``the electron is prepared in spin up along $x$ and also along $z$'' are impossible events.}.
 \end{defn}
 
 \begin{defn}[Preparation independence]\label{def_indep}
  Two systems are said independent iff the preparation of one does not affect the preparation of the other. Formally, all (pure) state pairs $(\pj{a}, \pj{b}) \in \pj{\mathcal{A}}\times \pj{\mathcal{B}}$ are compatible.
 \end{defn}
 
 \begin{prop}\label{prop_singleBorn}
  Given two systems, each prepared independently in their own state, the probability of measuring a value for one system depends only on the preparation of that system. That is, $P(\pj{a_1}|\pj{a_2}\wedge \pj{b})=P(\pj{a_1}|\pj{a_2})$.
 \end{prop}
 \begin{proof}
  We first note that, by postulate \ref{post_measurements}, the probability of measuring a value for one system depends only on the preparation of that system, which means that it is independent of the properties of any other system. Therefore $P(\pj{a_1} | \pj{a_2} \wedge \pj{b}) = \frac{\<a_1| a_2 \> \< a_2 |a_1\>}{\< a_1 | a_1 \>\< a_2 | a_2 \>} = P(\pj{a_1} | \pj{a_2})$
 \end{proof}
 
 \begin{defn}[Composite systems]\label{def_comp}
  Let A and B be two systems. The composite system C of A and B is formed by the simple collection of those and only those two systems, in the sense that it satisfies the following two requirements.
  \begin{enumerate}
   \item Every preparation of both subsystems is a preparation of the
   composite. Formally, let $\pj{\mathcal{C}}$ be the state space
   for C, there exists a map (not yet specified)
   $M:\pj{\mathcal{A}}\times\pj{\mathcal{B}}\to\pj{\mathcal{C}}$
   such that, for any compatible pair of (pure) states $(\pj{a},\pj{b}) \in \pj{\mathcal{A}}\times\pj{\mathcal{B}}$, the proposition $\pj{a} \wedge \pj{b}$ is equivalent to the (pure) state $M(\pj{a},\pj{b}) \in \pj{\mathcal{C}}$ where $M$ returns the state of the composite system where the subsystems were prepared in the given states. In other words, $\pj{a} \wedge \pj{b}$ and $M(\pj{a},\pj{b})$ correspond to the same event in probability space\footnote{We will end up proving that the map $M$ leads to the tensor product.}.
   \item For every preparation of the composite, local projective measurements must have at least one
   outcome with non-zero probability. Formally, for every $\pj{c} \in \pj{\mathcal{C}}$, we can find at least $\pj{a} \in \pj{\mathcal{A}}$ and $\pj{b} \in \pj{\mathcal{B}}$ such that $P(\pj{a} \wedge \pj{b}|\pj{c})\neq 0$. 
  \end{enumerate}
  It is important to understand that these requirements are necessary. Requirement 1 ensures that the composite system is well defined at least when the components are prepared independently. Conceptually, this ensures that the composite system contains all the properties of the components. Note that superselection rules or other restrictions may prevent the independent preparation of all possible pairs (e.g.~two fermions cannot be jointly prepared in the same state). The tensor product is recovered only when all pairs are compatible. Requirement 2 ensures that it does not contain properties that are orthogonal to all the components' properties, i.e.~that the composite system contains {\em only} the components. Violation of the second requirement would mean that some composite states would not define all the properties of the subsystems. That is, while the systems $A$ and $B$ by themselves would define observables $O_A$ and $O_B$, when grouped together all values would be assigned zero probability; those observables no longer exist. In this case, the nature of the system would have changed so radically we would no longer call it a composite system.
 \end{defn}
 
 \begin{prop}[Span surjectivity, H3]\label{prop_spanSurj}
  The map $M :
  \pj{\mathcal{A}}\times\pj{\mathcal{B}} \to \pj{\mathcal{C}}$ is span
  surjective, meaning that the span of the image coincides with the
  whole space. That is $Sp(\{ c \in \mathcal{C} \, | \, \pj{c} \in
  M(\pj{\mathcal{A}}, \pj{\mathcal{B}})\}) = \mathcal{C}$.
 \end{prop}
 \begin{proof}
  Consider $I=\{ c \in \mathcal{C} \, | \, \pj{c} \in M(\pj{\mathcal{A}}, \pj{\mathcal{B}})\}$ and its span. This forms a subspace of $\mathcal{C}$. By requirement 2 of \ref{def_comp}, for any $c \in \mathcal{C}$ we can always find $a \in \mathcal{A}$ and $b \in \mathcal{B}$ such that $P(\pj{a} \wedge \pj{b} | \pj{c} )\neq 0$. This means there is no element in $\mathcal{C}$ that is orthogonal to $Sp(I)$, therefore $Sp(I)$ must cover the whole $\mathcal{C}$.
 \end{proof}
 
 \begin{prop}[Totality, H1]\label{prop_totality}
  The map $M$ is in general a partial function.\footnote{A partial
   function is one that is not defined on the full domain. For
   example, $\sqrt(x)$ is a partial function since is not defined for
   $x<0$.} However, if A and B are independent, $M$ is a total function.\footnote{A total function is one that is defined on the full domain. For example, $x^2$ is a total function since it is defined for any $x$.}
 \end{prop}
 \begin{proof}
  As $M(\pj{a},\pj{b})$ is defined only if $(\pj{a},\pj{b}) \in \pj{\mathcal{A}}\times\pj{\mathcal{B}}$ are a compatible pair of pure states, it is not defined on pairs that are not compatible. If the two systems are independent, however, all pairs are allowed and $M$ is a total function.
 \end{proof}
 
 \begin{remark}
  As noted in \ref{def_comp}, if $\pj{a}$ and $\pj{b}$ are incompatible, $\pj{a} \wedge \pj{b}=\emptyset$ corresponds to the impossible event (i.e.~the empty set in the $\sigma$-algebra). This is not a state, and therefore $M(\pj{a},\pj{b})$ is not defined on incompatible pairs.
  
  However, in the end we will construct a map $m : \mathcal{A} \times \mathcal{B} \to \mathcal{C}$ on the vector spaces. There the zero vector plays the role of the impossible event. Therefore independent systems will map each pair to a non-zero element of the tensor product, while systems that are not independent will map incompatible states to the zero vector (e.g. the composite state of two electrons will exclude the cases where both electrons are in the same state).
 \end{remark}
 
 \begin{prop}[Statistical independence]\label{prop_statInd}
  Let $\pj{\mathcal{A}}$ and $\pj{\mathcal{B}}$ be the state spaces of two quantum systems and $\pj{\mathcal{C}}$ be the state space of their composite. The map $M : \pj{\mathcal{A}} \times \pj{\mathcal{B}} \to \pj{\mathcal{C}}$ is such that:
  \begin{align}
  P(M(\pj{a_1},\pj{b}) | M(\pj{a_2},\pj{b})) = P(\pj{a_1} | \pj{a_2})  \\ P(M(\pj{a},\pj{b_1}) | M(\pj{a},\pj{b_2})) = P(\pj{b_1} | \pj{b_2})
  \end{align}
  for all $a, a_1, a_2 \in \mathcal{A}$ and $b, b_1, b_2 \in \mathcal{B}$
 \end{prop}
 \begin{proof}
  By \ref{prop_singleBorn} we have $P(\pj{a_1} | \pj{a_2} \wedge \pj{b}) = P(\pj{a_1} | \pj{a_2})$ and similarly $P(\pj{b_1} | \pj{a} \wedge \pj{b_2}) = P(\pj{b_1} | \pj{b_2})$.  Using standard probability rules and remembering that $M(\pj{a},\pj{b})\equiv \pj{a} \wedge \pj{b}$ by \ref{def_comp}, we have $P(M(\pj{a_1},\pj{b}) | M(\pj{a_2},\pj{b})) = P(\pj{a_1} \wedge \pj{b} | \pj{a_2} \wedge \pj{b})  = P(\pj{b} | \pj{a_2} \wedge \pj{b}) P(\pj{a_1} | \pj{a_2} \wedge \pj{b} \wedge \pj{b})  = P(\pj{b} | \pj{b}) P(\pj{a_1} | \pj{a_2} \wedge \pj{b}) = P(\pj{a_1} | \pj{a_2})$, since trivially $P(\pj{b}|\pj{b})=1$. Similarly $P(M(\pj{a},\pj{b_1}) | M(\pj{a},\pj{b_2})) =P(\pj{a} \wedge \pj{b_1} | \pj{a} \wedge \pj{b_2}) = P(\pj{b_1} | \pj{b_2})$
 \end{proof}
 
 \begin{prop}[Fundamental theorem of projective geometry]\label{prop_fundProj}
  Let $X$ and $Y$ be two separable complex Hilbert spaces and $\pj{X}$ and $\pj{Y}$ their respective projective spaces. Let $M : \pj{X} \to \pj{Y}$ be a map such that $P(\pj{v}|\pj{w}) = P(M(\pj{v})|M(\pj{w}))$. Then we can find, up to a total phase, a unique map $m : X \to Y$ such that $\pj{m(v)}=M(\pj{v})$. Moreover, $m$ is either linear, $\<v|w\> = \<m(v)|m(w)\>$, or anti-linear, $\<v|w\> = \<m(w)|m(v)\>$.
 \end{prop}
 
 \begin{remark}
  The above proposition is, for the most part, an adaptation of the fundamental theorem of projective geometry \cite{fun}. The conservation of the probability imposes the semi-linear map to be either linear or anti-linear (i.e. conjugate-linear). This is not new, it is essentially Wigner's theorem, but the proof we offer is insightful as it clearly shows the connection between the construction of the map and the choice of gauge.
 \end{remark}
 
 \begin{proof}
  First we note that, given an orthonormal basis $\{e_i\}_{i \in I}$ over $X$, we can use $M$ to construct a corresponding basis over $Y' \subseteq Y$ where $\pj{Y'} = M(\pj{X})$. In fact, for each $\pj{e_i}$, pick a unit $u_i \in M(\pj{e_i})$. We have $\delta_{ij} = |\<e_i | e_j \>|^2 = P(e_i | e_j) = P(\pj{e_i}|\pj{e_j}) = P(M(\pj{e_i})|M(\pj{e_j})) = P(\pj{u_i}|\pj{u_j})= |\<u_i | u_j \>|^2$. The set $\{u_i\}_{i \in I}$ spans the entire $Y'$ since for all $y \in Y'$ we can find $x \in X$ and at least one $u_i$ such that $|\<y | u_i \>|^2 = P(\pj{y}|\pj{u_i}) = P(M(\pj{x})|M(\pj{e_i})) = P(x | e_i) = |\<x | e_i \>|^2 \neq 0$. Note that we have an arbitrary choice for each $u_i$, since we have to pick a vector from the unit circle (i.e. a phase for each basis vector). This corresponds to a choice of gauge.
  
  We also note that the map is colinear, meaning that if $U_X, V_X \subseteq X$ are two subspaces such that $U_X \subset V_X$, then $U_Y, V_Y \subseteq Y$ such that $\pj{U_Y} = M(\pj{U_X})$ and $\pj{V_Y} = M(\pj{V_X})$ are subspaces of $Y$ and $U_Y \subset V_Y$. In fact, take an orthonormal basis $\{e_i\}_{i \in I} \subset X$ such that $\{e_i\}_{i \in I_U \subset I} \subset \{e_i\}_{i \in I_V \subset I}$ are bases for $U_X$ and $V_X$ respectively. An element of $X$ belongs to $U_X$ if and only if it is not orthogonal only to elements of the basis of $U_X$ and belongs to $V_X$ only if it not orthogonal only to elements of the basis of $V_X$. As the map $M$ preserves orthogonality, these relationships are preserved by the map. Therefore $U_Y$ and $V_Y$ are subspaces of $Y$ such that $U_Y \subset V_Y$.
  
  Additionally we note that, for any colinear map, given two subspaces $U_1, U_2 \subseteq X$ we have $M(\pj{Sp(U_1, U_2)}) = \pj{Sp(M(\pj{U_1}), M(\pj{U_2}))}$. In fact, $\pj{Sp(U_1,U_2)}$ is the smallest subspace containing all vectors in $\pj{U_1}$ and $\pj{U_2}$. In the same way, $\pj{Sp(M(U_1),M(U_2))}$ is the smallest subspace containing all vectors in $M(\pj{U_1})$ and $M(\pj{U_2})$. Since the subspace inclusion is preserved by $M$, we must have $M(\pj{Sp(U_1, U_2)}) = \pj{Sp(M(\pj{U_1}), M(\pj{U_2}))}$.
  
  We now use the gauge freedom to redefine the basis such that for all $i$ we have $M(\pj{e_i}) = \pj{v_i}$ and $M(\pj{e_1 + e_i}) = \pj{v_1 + v_i}$. Let $v_1 = u_1$. This is the only arbitrary choice we make, and corresponds to the choice of a global phase. For each $i>1$, consider $e_1 + e_i$. This will belong to the subspace $Sp(e_1, e_i)$. This subspace, when mapped through $M$, will give us the subspace spanned by $v_1$ and $u_i$. That is, $M(\pj{Sp(e_1, e_i)})=\pj{Sp(v_1, u_i)}$. This means we can find a unique $k \in \mathbb{C}$ such that $M(\pj{e_1+e_i})=\pj{v_1+ku_i}$. We fix $v_i = k u_i$. Note that $P(\pj{e_1} | \pj{e_1+e_i}) = \frac{1}{2}=  P(\pj{e_i} | \pj{e_1+e_i})=P(\pj{v_1} | \pj{v_1+k u_i})=P(\pj{u_i} | \pj{v_1+k u_i})$. Therefore $|k| = 1$ and $k u_i = v_i$ is a unit vector.
  
  Now we want to show that $M(\pj{e_1 + c e_i})=\pj{v_1 + \tau_i(c)
   v_i}$ where either $\tau_i(c) = c$ or $\tau_i(c) = c^\dagger$. For
  each $i$, consider $w = e_1 + c e_i \in Sp(e_1, e_i)$. Since
  $M(\pj{w}) \subset \pj{Sp(v_1, v_i)}$, there must be a $\tau_i(c)$
  such that $\pj{v_1 + \tau_i(c) v_i} = M(\pj{w})$. Since we must have
  $P(\pj{e_i} | \pj{w}) = P(\pj{v_i} | M(\pj{w}))$ and $P(\pj{e_1 +
   e_i} | \pj{w}) = P(\pj{v_1 + v_i} | M(\pj{w}))$, we must have
  $|c| = |\tau_i(c)|$ and $cos(arg(c)) = cos(arg(\tau_i(c)))$ for any
  $c$. This means that either $\tau_i(c) = c$ or $\tau_i(c) =
  c^*$.
  
  Next we want to show that $\tau_i(c) = \tau_j(c)$ for all pairs $(i,j)$. That is, either we have to take the complex conjugate of all components or of none. Consider $e_i - e_j$. We have $\pj{e_i - e_j} \subset \pj{Sp(e_i, e_j)}$ and, for any $c \in \mathbb{C}$, $\pj{e_i - e_j} \subset \pj{Sp(e_1 + c e_i, e_1 + c e_j)}$. By construction, we have $M(\pj{e_i - e_j})\subset \pj{Sp(v_i, v_j)}$ and $M(\pj{e_i - e_j})\subset \pj{Sp(v_1 + \tau_i(c) v_i, v_1 + \tau_j(c) v_j)}$. Therefore $M(\pj{e_i - e_j}) = \pj{Sp(v_i, v_j)} \cap \pj{Sp(v_1 + \tau_i(c) v_i, v_1 + \tau_j(c) v_j)} = \pj{\tau_i(c) v_i - \tau_j(c) v_j}$. This means that, for all $c$, $\tau_i(c) = \tau_j(c)$.
  
  Now we show that for all $c_2, ..., c_n \in \mathbb{C}$ we have $M(\pj{e_1 + c_2 e_2 + ... + c_n e_n}) = \pj{v_1 + \tau(c_2) v_2 + ... + \tau(c_n) v_n}$. We prove this by induction. If only the first two components are non-zero, we have  $M(\pj{e_1 + c_2 e_2}) = \pj{v_1 + \tau(c_2) v_2}$ by construction. Let $2 < p \leq n$. If we assume $M(\pj{e_1 + c_2 e_2 + ... + c_{p-1} e_{p-1}}) = \pj{v_1 + \tau(c_2) v_2 + ... + \tau(c_{p-1}) v_{p-1}}$, then $M(\pj{e_1 + c_2 e_2 + ... + c_p e_p}) \subset M(\pj{Sp(e_1 + c_2 e_2 + ... + c_{p-1} e_{p-1}, e_p})) = \pj{Sp(v_1 + c_2 v_2 + ... + c_{p-1} v_{p-1}, v_p)}$. This means that there exists $k_p \in \mathbb{C}$ such that $M(\pj{e_1 + c_2 e_2 + ... + c_p e_p}) = \pj{v_1 + c_2 v_2 + ... + c_{p-1} v_{p-1} + k_p v_p)}$. But we also have $M(\pj{e_1 + c_2 e_2 + ... + c_p e_p}) \subset M(\pj{Sp(e_1 + c_p e_p, c_2 e_2 + ... + c_{p-1} e_{p-1}})) = \pj{Sp(v_1 + \tau(c_p) v_p, c_2 v_2 + ... + c_{p-1} v_{p-1}, v_p)}$. The only way this can work is if $k_p = \tau(c_p)$.
  
  The above works also over a countable sum. That is, for all $c_2, ..., c_n, ... \in \mathbb{C}$ we have $M(\pj{e_1 + c_2 e_2 + ... + c_n e_n + ...}) = \pj{v_1 + \tau(c_2) v_2 + ... + \tau(c_n) v_n + ...}$. Let $X$ be a separable space. Let $a = \sum_{k=1}^\infty c_k e_k$ such that $c_1 = 1$. Let $a_i = \sum_{k=1}^i c_k e_k$ be the sum of the first $i$ components. We have $\lim\limits_{i \to \infty} a_i = a$. Let $b_i = \sum_{k=1}^i \tau(c_k) v_k$ and $b = \sum_{k=1}^\infty \tau(c_k) e_k$. We have $\lim\limits_{i \to \infty} b_i = b$. We already know that $\pj{b_i} = M(\pj{a_i})$ for all finite $i$. We need to show that $\pj{b} = M(\pj{a})$. First note that, given $a,b \in Y$, $a=b$ if and only if $\<a,c\> = \<b,c\>$ for all $c \in Y$. Therefore $\pj{a}=\pj{b}$ if and only if $P(\pj{a}, \pj{c}) = P(\pj{b}, \pj{c})$ for all $c \in Y$. For all $c \in X$ we have $\lim\limits_{i \to \infty} P(\pj{a_i}, \pj{c}) = P(\pj{a}, \pj{c}) = P(M(\pj{a}), M(\pj{c}))$ we also have  $\lim\limits_{i \to \infty} P(\pj{a_i}, \pj{c}) = \lim\limits_{i \to \infty} P(M(\pj{a_i}), M(\pj{c})) = \lim\limits_{i \to \infty} P(b_i, M(\pj{c})) = P(\pj{b}, M(\pj{c}))$. Therefore $P(M(\pj{a}), M(\pj{c})) = P(\pj{b}, M(\pj{c}))$ for all $c \in X$. Note that $M$ is bijective over $Y'=M(\pj{X})$. Therefore $P(M(\pj{a}), \pj{c}) = P(\pj{b}, \pj{c})$ for all $c \in Y$ and $\pj{b} = M(\pj{a})$.
  
  We also need to show the above works when there is no component on the first element of the basis. That is, for all $c_2, ..., c_n \in \mathbb{C}$ we have $M(\underbar{$c_2 e_2 + ... + c_n e_n$}) = \underbar{$\tau(c_2) v_2 + ... + \tau(c_n) v_n$}$. First note that $M(\underbar{$c_2 e_2 + ... + c_n e_n$}) \subset M(\underbar{$Sp(e_2, ..., e_n)$)} = \underbar{$Sp(v_2, ..., v_n)$}$. Also note that $M(\underbar{$c_2 e_2 + ... + c_n e_n$}) \subset M(\underbar{$Sp(e_1, e_1 + c_2 e_2 + ... + c_n e_n)$)} = \underbar{$Sp(v_1, v_1 + \tau(c_2) v_2 + ... + \tau(c_n) v_n)$}$. The only way this can work is if $M(\underbar{$c_2 e_2 + ... + c_n e_n$}) = \underbar{$\tau(c_2) v_2 + ... + \tau(c_n) v_n$}$. With same reasoning as before, we can extend the sum to the countably infinite case.
  
  We can now define $m : X \to Y$ such that $m(e_i) = v_i$ for all $i$ and $m(\sum_{i \in I} c_i e_i) = \sum_{i \in I} \tau(c_i) v_i$. This means $\underbar{$m(\sum_{i \in I} c_i e_i)$} = M(\underbar{$\sum_{i \in I} c_i e_i$})$. Moreover, if $\tau(c) = c$ we have $\<m(\sum_{i \in I} c_i e_i)|m(\sum_{j \in I} d_j e_j)\> = \<\sum_{i \in I} c_i v_i|\sum_{j \in I} d_j v_j\> = c_i^* d_j \delta_{ij} = \<\sum_{i \in I} c_i e_i|\sum_{j \in I} d_j e_j\>$. On the other hand, if $\tau(c) = c^*$ we have $\<m(\sum_{i \in I} c_i e_i)|m(\sum_{j \in I} d_j e_j)\> = \<\sum_{i \in I} c_i^* v_i|\sum_{j \in I} d_j^* v_j\> = \<\sum_{j \in I} d_j v_j|\sum_{i \in I} c_i v_i\> = d_i^* c_j \delta_{ij} = \<\sum_{j \in I} d_j e_j|\sum_{i \in I} c_i e_i\>$. This can be extended to the case where the basis is countable.
 \end{proof}
 
 \begin{remark}
  The fact that the proposition identifies either a linear map or an anti-linear (i.e.~conjugate-linear) corresponds, in physics terms, to a choice of convention. As analogies: a change in metric signature in relativity would change the mathematical space but not the physics; in classical phase-space, a change in signature of the symplectic form would change the mathematical space, but not the physics it represents. These choices are widely recognized as a matter of personal preference.
  
  In simple terms, for a Hilbert space, the conjugate
  vector space is equivalent to the dual space, so we
  could equivalently choose one or the other. For
  example, Schr\"odinger, in the papers in which he
  introduces the Schr\"odinger equation, writes it
  with both signs, as the choice of sign of the imaginary part of the wave function is arbitrary: one sign refers to
  the Hilbert space, the other to the dual space, namely
  to the complex-conjugate wave function \cite{sch}. So we can think of the anti-linear map as one that preserves the inner product but maps ket vectors into bra vectors. Looking ahead, the above result does not exclude a composition map similar to the tensor product, but that maps the kets of one or both subsystems into bras in the composite system. This would only make the representation of the composite physical system more complicated, as we need to keep track of the different conventions in the different subspaces. Therefore, without changing the physics, we can always mathematically redefine the second space so that the resulting map is linear. With this in mind, we will assume that the map between the spaces is linear, which will in turn lead to identifying the tensor product as a unique composition map.
  
  Another way to look at this is that Hermitian operators, and therefore all 
  the physics, are invariant under an anti-linear transformation.
  In contrast, anti-Hermitian operators will change sign. This 
  changes the connection between the generators and the generated 
  transformations (i.e.~while $A$ generates $exp(\frac{Ab}{\imath\hbar})$ 
  on one space, the mapped $A$ will generate 
  $exp(-\frac{Ab}{\imath\hbar})$ in the mapped space). Note that the 
  choice of whether to put the minus or not is arbitrary as long as one is consistent across all generators and transformations.
  Similarly, we typically define $[A, B] = AB - BA$ but we could have 
  alternatively chosen $[A, B] = BA - AB$. The anti-linear map is simply 
  a change of that convention.
  
  Note that this unnecessary subtlety could in principle be avoided by
  reformulating quantum mechanics in terms of quantum states given by
  density matrices $\rho=|\psi\>\<\psi|$ (which contain both  kets and
  bras), as is done, for example in \cite{ozawa,holevo}. In this paper
  we employed the more familiar formulation in which quantum states are
  rays in Hilbert space (identified either by kets or bras).
 \end{remark}
 
 \begin{prop}[Bilinearity, H2]\label{prop_bilinearity}
  Given $M$ in \ref{def_comp}, if we can find an $m : \mathcal{A} \times \mathcal{B} \to \mathcal{C}$ such that for all $(a,b) \in \mathcal{A} \times \mathcal{B}$, we have $\pj{m(a,b)} = M(\pj{a}, \pj{b})$, then $m$ must be bilinear. That is:
  \begin{align}
  m(k_1a_1 + k_2a_2, b)=k_1m(a_1, b) + k_2m(a_2, b) \\
  m(a, k_1b_1 + k_2b_2)=k_1m(a, b_1) + k_2m(a, b_2)
  \end{align}
  for all $a, a_1, a_2 \in \mathcal{A}$, $b, b_1, b_2 \in \mathcal{B}$ and $k_1, k_2 \in \mathbb{C}$.
 \end{prop}
 
 \begin{proof}
  If we fix $b \in \mathcal{B}$, then we have $M_b : \pj{\mathcal{A}} \to \pj{\mathcal{C}}$ where $M_b(\pj{a}) = M(\pj{a}, \pj{b})$. By \ref{prop_statInd} and \ref{prop_fundProj} we can find a linear map $m_b : \mathcal{A} \to \mathcal{C}$ such that $\pj{m_b(a)} = M_b(\pj{a}) = M(\pj{a}, \pj{b})$. As this must map subspace to subspace, we must have $m(a, b) = k m_b(a)$ for some $k \in \mathbb{C}$. Since $m_b$ is linear, we have $m(k_1a_1 + k_2a_2, b)=k_1m(a_1, b) + k_2m(a_2, b)$ for any $a_1, a_2 \in \mathcal{A}$ and $k_1, k_2 \in \mathbb{C}$. We can repeat the argument fixing $a \in \mathcal{A}$, and find $m(a, k_1b_1 + k_2b_2)=k_1m(a, b_1) + k_2m(a, b_2)$ for any $b_1, b_2 \in \mathcal{B}$ and $k_1, k_2 \in \mathbb{C}$.
 \end{proof}
 
 \begin{prop}[Subsystems' basis gives composite system
  basis]\label{prop_basis}
  Let $\{a_i\}_{i\in I}$ and $\{b_j\}_{j \in J}$ be bases of $\mathcal{A}$ and $\mathcal{B}$ respectively, then a set of unit vectors $\{e_{ij}\}_{(i,j) \in I \times J} \subset \mathcal{C}$  such that $e_{ij} \in M(\underbar{$a_i$}, \underbar{$b_j$})$ forms a basis for $\mathcal{C}$.
 \end{prop}
 
 \begin{proof}
  
  Since $M$ is a map on the projective spaces, it maps spans to spans.
  Since the span of the basis of $\mathcal{A}$ and $\mathcal{B}$ is
  the whole space, then the span of the image of the basis is the
  whole image of $M$. By \ref{prop_spanSurj}, the image of $M$
  coincides with the whole $\mathcal{C}$. Therefore, given
  $\{a_i\}_{i\in I}$ and $\{b_j\}_{j \in J}$ bases of $\mathcal{A}$
  and $\mathcal{B}$ respectively, any set of unit vectors
  $\{e_{ij}\}_{(i,j) \in I \times J} \subset \mathcal{C}$ such that
  $e_{ij} \in M(\pj{a_i}, \pj{b_j})$ spans the whole $\mathcal{C}$.
  
  Now consider $P(M(\pj{a_i},\pj{b_j})|M(\pj{a_k},\pj{b_l}))$. If $i=k$ and $j=l$ we have $P(M(\pj{a_i},\pj{b_j})|M(\pj{a_k},\pj{b_l}))= P(M(\pj{a_i},\pj{b_j})|M(\pj{a_i},\pj{b_j})) = 1$. If $i\neq k$, we have $P(M(\pj{a_i},\pj{b_j})|M(\pj{a_k},\pj{b_l})) = P(\pj{a_i} \wedge \pj{b_j} | \pj{a_k} \wedge \pj{b_l}) \leq P(\pj{a_i} | \pj{a_k} \wedge \pj{b_l})$. By \ref{prop_singleBorn} we have $P(\pj{a_i} | \pj{a_k} \wedge \pj{b_l}) = P(\pj{a_i} | \pj{a_k}) = 0$ since $a_i$ and $a_k$ are different elements of an orthogonal basis. Therefore we have $P(M(\pj{a_i},\pj{b_j})|M(\pj{a_k},\pj{b_l}))=\delta_{ik}\delta_{jl}$ which means $\<e_{ij} | e_{kl} \> = \delta_{ik}\delta_{jl}$.
  
  The elements $e_{ij}$ form a set of orthonormal vectors that span the whole space and are therefore a basis.
 \end{proof}
 
 \begin{thrm}[Composite system theorem]\label{theo}
  The state space of a composite system of independent systems is given by the tensor product of the spaces of the component systems.
 \end{thrm}
 
 \begin{proof}
  We are looking for a map $m : \mathcal{A} \times \mathcal{B} \to \mathcal{C}$ such that, for all $(a,b) \in \mathcal{A} \times \mathcal{B}$ we have $\pj{m(a,b)} = M(\pj{a}, \pj{b})$. We saw in \ref{prop_bilinearity} that if $m$ exists, it must be bilinear.
  
  Now we show that, if $m$ exists, then $\mathcal{C} \cong \mathcal{A}
  \otimes \mathcal{B}$ (where $\cong$ indicates an isomorphism) and $m :
  \mathcal{A} \times \mathcal{B} \to \mathcal{A} \otimes \mathcal{B}$ is
  the standard map from the Cartesian product to the tensor product. As $m
  : \mathcal{A} \times \mathcal{B} \to \mathcal{C}$ is a bilinear
  operator, by the universal property of the tensor product we can find
  a linear operator $\hat{m} : \mathcal{A} \otimes \mathcal{B} \to
  \mathcal{C}$ such that $m(a, b) = \hat{m}(a \otimes b)$. By
  \ref{prop_basis} the set $\{ m(a_i, b_j)\}_{(i,j) \in I \times J}$
  forms a basis since $ m(a_i, b_j) \in M(\pj{a_i}, \pj{b_j})$ for all
  $(i, j)$, therefore $\hat{m}(\{ a_i \otimes b_j\}_{(i,j) \in I \times
   J})$ also forms a basis since $\hat{m}(a_i\otimes b_j)=m(a_i, b_j)$.
  By \ref{prop_totality}, each $m(a_i,b_j)$ will correspond to a unit
  vector in $\mathcal{C}$. We have $\<\hat m(a_i\otimes b_j)| \hat
  m(a_k\otimes b_\ell)\>_\mathcal{C} =\<m(a_i, b_j)| m(a_k,
  b_\ell)\>_\mathcal{C} = \delta_{ik}\delta_{j\ell} = \<a_i\otimes b_j|
  a_k \otimes b_\ell\>_{\otimes}$. The function $\hat{m}$, then,
  preserves the inner product across all elements of the basis and is
  therefore an isomorphism for Hilbert spaces. We have $\mathcal{C}
  \cong \mathcal{A} \otimes \mathcal{B}$ and $m(a, b) = \hat{m}(a
  \otimes b) \cong a \otimes b$.
  
  Given that the tensor product map exists and it satisfies all the properties $m$ must satisfy, then $m$ exists and it is the tensor product.
 \end{proof}

 In conclusion, as an aside, we note that in quantum field theory one
 tends to avoid problems connected with tensor products of infinite
 dimensional spaces by focusing on algebraic commutation structures,
 e.g.~\cite{giddins,roos}.  In particular, the recent MIP*=RE result
 \cite{mipre} implies that, in infinite dimensions, the tensor product
 is strictly less computationally powerful than the commutation
 structures, emphasizing the difference among these two structures, at
 least for the infinite-dimensional case. 

 Moreover, we note that in our paper we mainly focused on systems
 where no superselection rules or other restrictions to the state
 space are present: it is possible to prepare each subsystem of a
 composite system in a state that is independent of the other systems
 (preparation independence).  This is the only case in which the
 tensor product can be properly employed: the Hilbert space of
 composite systems that have restrictions is {\em not} the tensor
 product of the component spaces, but a subspace of it (e.g.~the
 anti-symmetric subspace for fermions).  Typically this is ignored in
 the literature, since the tensor product formalism is very convenient
 and is often used also in these cases, and superselection rules are
 typically avoidable \cite{susskind,zanardi,zanardilloyd}. A typical
 example \cite{tellerbook} comes from quantum field theory.  It is
 customary in basically all quantum optics literature to treat
 different modes of the radiation field (e.g.~the output of two
 lasers) as independent systems composed through the tensor product.
 Clearly the electromagnetic field is a single system and an agent who
 is able to access an optical mode that is a linear combination of the
 two will give a quantum description for it that cannot easily
 accommodate tensor products. Similarly, an agent can consider two
 electrons as two systems, joined with the tensor product, whenever
 they are distinguishable for all practical purposes (e.g.~the
 electrons are in widely separated physical locations). Yet, in
 principle, electrons are just excitations of a field, and the `true'
 quantum system is the field, not the single electrons
 \cite{teller,tellerbook}.  So, in quantum field theory, the quantum
 systems that should be joined through tensor products are the
 different fields and {\em not} the particles, which are just
 excitations (states) of the fields. In the words of Teller
 (\cite{tellerbook}, pg.22), tensor products can be safely used only
 if there is a ``primitive thisness'', which is captured in the
 definition of system.

\vskip 1\baselineskip
 To conclude, we give a schematic outline of the logical implications
 that led us to the result. This is an expanded version of Fig.~1 of
 the main paper:
 \begin{enumerate}
  \item P\ref{post_state}: states and observables postulate.
  
  \item P\ref{post_measurements}: Born rule (measurement postulate). 
  
  \item Def \ref{def_indep}: Preparation independence: systems are independent if the preparation of one does not affect the other.
  
  \item P\ref{post_measurements} $\Rightarrow$ \ref{prop_singleBorn}: the outcome probabilities depend only on the
  inner product.
  
  \item Def \ref{def_comp}: Composite system definition: A composite system is a
  collection of the subsystems (i.e.~all compatible states give a preparation) and only of the subsystems
  (i.e.~all composite preparations give non-trivial measurements on the
  subsystems).

  \item P\ref{post_measurements} + Def \ref{def_comp} $\Rightarrow$ \ref{prop_spanSurj} (H3): Span
  surjectivity (all composiste $\cal C$ are superpositions of ${\cal A}$ and ${\cal B}$).
  
  \item Def \ref{def_indep} + Def \ref{def_comp} $\Rightarrow$ \ref{prop_totality} (H1):
  Totality (all possible state pairs of the subsystems correspond to a state of the composite). 
  
  \item P\ref{post_state} + \ref{prop_singleBorn} +  Def \ref{def_comp} $\Rightarrow$ \ref{prop_statInd}: Statistical
  independence (if one subsystem does not change, the probability on the composite system is given by the probability of the subsystem that changes).
  
  \item \ref{prop_fundProj}: Fundamental theorem of projective geometry (preserving square of inner product leads to unique linear map)
  
  \item \ref{prop_statInd} + \ref{prop_fundProj} $\Rightarrow$ \ref{prop_bilinearity} (H2) composition map on vector spaces is bilinear.
  
  \item P\ref{post_measurements} + Def \ref{def_comp} + \ref{prop_spanSurj} (H3) $\Rightarrow$ \ref{prop_basis}: Basis
  carries over from subsystems to composite
  
  \item \ref{prop_totality} + \ref{prop_bilinearity} + \ref{prop_basis} $\Rightarrow$ \ref{theo}: the composition map is the tensor product.
  
 \end{enumerate}

 \subsection{Addendum}
 
 The above proof relies only on the independent preparation of subsystems and not their measurements. However, during the review process, an anonymous referee contributed a sketch for a proof that shows very directly that the state and measurement postulates also imply independence of measurements. The insight is that the mixture created by all possible measurement outcomes on $B$ must behave overall as a pure state on $A$. Since pure states are extreme points, this can only happen if every measurement outcome on $B$ leaves $A$ in a pure state, which makes the probability factorize.
 
 \begin{prop}
  Let $\pj{\psi}$ and $\pj{\phi}$ be two preparations for $A$ and $B$. Let $\pj{a}$ and $\pj{b}$ be two measurements on the respective systems. Then $P(\pj{a} \wedge \pj{b} | \pj{\psi} \wedge \pj{\phi}) = P(\pj{a} | \pj{\psi}) P(\pj{b} | \pj{\phi})$.
 \end{prop}
 \begin{proof}
  Consider $P(\pj{a} \wedge \pj{b} | \pj{\psi} \wedge \pj{\phi})$. We can imagine first performing the measurement on $B$ and then conditioning the result on $A$. We have $P(\pj{a} \wedge \pj{b} | \pj{\psi} \wedge \pj{\phi}) = P(\pj{a} | \pj{b}, \pj{\psi} \wedge \pj{\phi}) P(\pj{b} | \pj{\psi} \wedge \pj{\phi})$ where by ``$\pj{b}, \pj{\psi} \wedge \pj{\phi}$'' we mean that the systems were prepared in $\pj{\psi}$ and $\pj{\phi}$ respectively and $\pj{b}$ was measured. Any preparation on $A$ can be expressed with a mixed state, and any measurement on $B$ depends only on $\phi$, so we have $P(\pj{a} | \pj{b}, \pj{\psi} \wedge \pj{\phi}) P(\pj{b} | \pj{\psi} \wedge \pj{\phi}) = P(\pj{a} | \rho(\pj{b}, \pj{\psi}, \pj{\phi})) P(\pj{b} | \pj{\phi})$. We also must have $\sum_b P(\pj{a} \wedge \pj{b} | \pj{\psi} \wedge \pj{\phi}) = P(\pj{a} | \pj{\psi} \wedge \pj{\phi}) = P(\pj{a} | \pj{\psi})$.
  
  Putting it all together, we have $P(\pj{a} | \pj{\psi}) = \sum_b P(\pj{a} | \rho(\pj{b}, \pj{\psi}, \pj{\phi})) P(\pj{b} | \pj{\phi})$. Which means $|\psi\>\<\psi| = \sum_b P(\pj{b} | \pj{\phi}) \rho(\pj{b}, \pj{\psi}, \pj{\phi})$. But the only way that a mixture of mixed states can be equal to a pure state is if all mixed states are the same pure state. Therefore $\rho(\pj{b}, \pj{\psi}, \pj{\phi}) = |\psi\>\<\psi|$ for all $b$ and $\phi$. We finally have $P(\pj{a} \wedge \pj{b} | \pj{\psi} \wedge \pj{\phi}) = P(\pj{a} | \rho(\pj{b}, \pj{\psi}, \pj{\phi})) P(\pj{b} | \pj{\phi}) = P(\pj{a} | \pj{\psi}) P(\pj{b} | \pj{\phi})$.
 \end{proof}

\end{document}